\documentclass[11pt]{article}
\usepackage[utf8]{inputenc}
\usepackage[left=1in, right=1in, bottom=1.15in, top=1.1in]{geometry}

\usepackage{tingfeng}

\title{Computing a Fixed Point of Contraction Maps\\ in Polynomial Queries}

\author{
Xi Chen\footnote{Supported by NSF grants IIS-1838154, CCF-2106429 and CCF-2107187.}\\ Columbia University\\\url{xichen@cs.columbia.edu}\hspace{-0.2cm}
\and Yuhao Li\footnote{Supported by NSF grants IIS-1838154, CCF-2106429 and CCF-2107187.}\\Columbia University\\\url{yuhaoli@cs.columbia.edu}\hspace{-0.2cm}
\and Mihalis Yannakakis\footnote{Supported by NSF grants CCF-2107187, CCF-2212233, and CCF-2332922.}\\Columbia University\\\url{mihalis@cs.columbia.edu}\vspace{0.1cm}
}

\date{}
\begin{document}
\maketitle
\begin{abstract}
	We give an algorithm for finding an $\eps$-fixed point of a contraction map $f:[0,1]^k\mapsto[0,1]^k$ under the $\ell_\infty$-norm with query complexity $O (k\log (1/\eps ) )$. 
\end{abstract}

\thispagestyle{empty}
\newpage 
\setcounter{page}{1}
\newpage

\section{Introduction}
A map $f:\calM\mapsto\calM$ on a metric space $(\calM,d)$ is called a \emph{contraction} map (or a $(1-\gamma)$-contraction map) if there exists $\gamma\in(0,1]$ such that $d(f(x),f(y))\leq (1-\gamma)\cdot d(x,y)$ for all points $x,y\in \calM$.
In 1922, Banach~\cite{Ban1922} proved a seminal fixed point theorem which states that every contraction map must have a unique fixed point, i.e., there is a unique $x\in\calM$ that satisfies $f(x)=x$. Distinct from another renowned fixed point theorem by Brouwer, Banach's theorem not only guarantees~the uniqueness of the fixed point  but also provides a method for finding it: iteratively applying the map $f$ starting from any initial point will always converge to the unique fixed point. Over the past century, Banach's fixed point theorem has found extensive applications in many fields. For example, in mathematics it can be used to prove theorems such as the Picard–Lindelöf (or Cauchy-Lipschitz) theorem on the existence and uniqueness of solutions to differential equations (see e.g. \cite{CL55}), 
and the Nash embedding theorem~\cite{nash1956imbedding,gunther1989}. In optimization and machine learning, it is used in the convergence and uniqueness analysis of value and policy iteration in Markov decision processes and reinforcement learning~\cite{bellman1957markovian,howard1960dynamic}. Indeed, as pointed out by Denardo \cite{De67}, contraction mappings underlie many classical dynamic programming (DP) problems and sequential decision processes, including DP models of Bellman, Howard, Blackwell, Karlin and others.

A particularly important metric space to study the problem of finding a Banach's fixed point is the $k$-cube $[0,1]^k$ with respect to the $\ell_\infty$-norm,
  since many important problems can be reduced to that of finding an 
  $\eps$-fixed point (i.e., $x\in [0,1]^k$ satisfying $\|f(x)-x\|_\infty\le\eps$) in a ${(1-\gamma)}$-contraction map under the $\ell_\infty$-norm.
Such problems arise from a variety of fields including stochastic analysis, optimization, verification, semantics, and game theory.
For example, the classical dynamic programming models mentioned above
(Markov decision processes etc.) involve contraction maps under the $\ell_\infty$-norm. Furthermore, the same holds for several well-known open problems that 
have been studied extensively and are currently not known to be in $\P$.
For instance, Condon's simple stochastic games (SSGs) \cite{Con92} can be reduced to the problem of finding an $\eps$-fixed point in a $(1-\gamma)$-contraction map  over $[0,1]^k$ under the $\ell_\infty$-norm. 
A similar reduction from~\cite{EY10} extends to an even broader class of games, namely, Shapley's stochastic games~\cite{Shapley53}, which lay the foundation of multi-agent reinforcement learning~\cite{littman1994markov}.
The same holds also of course for other problems 
known to be subsumed by SSGs, like parity games, which are important in verification (see e.g. \cite{EJ91,CJKLS22}), and mean payoff games \cite{ZP96}. 
Crucially, in all these reductions, 
  both the approximation parameter $\eps$ and the contraction parameter $\gamma$ are inversely exponential in the input size. 
  Therefore, efficient algorithms in this context are those with a complexity upper bound that is polynomial in $k,\log(1/\eps)$ and $\log(1/\gamma)$.

In this paper we consider general algorithms that access the contraction map in a black-box manner (as an oracle), and study the \emph{query complexity} of finding an $\eps$-fixed point of a $(1-\gamma)$-contraction map over the $k$-cube $[0,1]^k$ under the $\ell_\infty$-norm
    (which we denote by $\Contraction(\eps,\gamma,k)$).
 An algorithm under this model is given $k$, $\eps$, $\gamma$, and oracle access to an unknown $(1-\gamma)$-contraction map $f$ over $[0,1]^k$. In each round the algorithm can send a point $x\in[0,1]^k$ to the oracle to reveal its value $f(x)$. The goal of the algorithm is to find an $\eps$-fixed point with as few queries as possible. 

\vspace{0.2cm}
\noindent \textbf{Prior work.} 
Despite much ongoing interest on this problem (e.g., \cite{EY10,DP11,DTZ18,FGMS20,Hol21,FGHS23}), 
  progress in understanding the query complexity of $\Contraction(\eps,\gamma,k)$ has been slow.
Banach's value iteration method needs $\Omega((1/\gamma)\log(1/\eps))$ iterations
  to converge to an~$\eps$-fixed point.
For the special case of $k=2$, 
  \cite{SS02} obtained an $O(\log ( {1}/{\eps} ) )$-query algorithm.
Subsequently, \cite{SS03} obtained an $O (\log^k ( {1}/{\eps}))$-query
algorithm for general $k$ by applying a nontrivial recursive binary search procedure across all $k$ dimensions.
(Recently \cite{FGMS20} obtained similar upper bounds for all $\ell_p$-norms with $2< p < \infty$, though the complexity grows to infinity as $p \rightarrow \infty$.)
Note, however, that all known upper bounds so far are exponential in either $k$ or $\log (1/\gamma)$,
  and this is in sharp contrast with the $\ell_2$-norm case, for which
  \cite{STW93,HKS99} gave an algorithm with both query and time complexity polynomial in
  $k,\log(1/\eps)$ and $\log(1/\gamma)$.

\vspace{0.2cm}
\noindent\textbf{Our contribution.} We obtain the first algorithm
  for $\Contraction(\eps,\gamma,k)$ with polynomial query complexity:
  
\begin{theorem}\label{theorem: main}
There is an $O (k\log ({1}/{\eps}))$-query algorithm for	$\Contraction(\eps,\gamma,k)$.
\end{theorem}

The observation below explains why our upper bound does not depend on $\gamma$:

\begin{observation}\label{firstob}
Let $f:[0,1]^k\mapsto [0,1]^k$ be a $(1-\gamma)$-contraction map under the $\ell_\infty$-norm.
Consider the map $g:[0,1]^k\mapsto[0,1]^k$ defined as $g(x):=(1-\eps/2)f(x)$. Clearly $g$ is a $(1-\eps/2)$-contraction. Let $x$ be any point with $\norm{g(x)-x}\leq \eps/2$. We have 
$$
\eps/2\ge \|g(x)-x\|_\infty =\|(1-\eps/2)f(x)-x\|_\infty\ge \|f(x)-x\|_\infty-\eps/2.
$$
This gives a black-box reduction 
  from $\Contraction(\eps,\gamma,k)$ to $\Contraction(\eps/2,\eps/2,k)$, which is both query-efficient and time-efficient.
\end{observation}

In \Cref{sec:mainalg}, we give an $O(k\log (1/\eps))$-query algorithm for 
  $\Contraction(\eps/2,\eps/2,k)$, from which \Cref{theorem: main} follows.
Indeed, note that Observation~\ref{firstob} holds even if $f$ is a \emph{non-expansive} map
  (i.e., $f$ has Lipschitz constant $1$:
  $\|f(x)-f(y)\|_\infty \le \|x-y\|_\infty$ for all $x,y\in [0,1]^k$).
As a result, the same query upper bound applies to $\NonExp(\eps,k)$, the 
  problem of finding an $\eps$-fixed point in a non-expansive map over $[0,1]^k$ under the $\ell_\infty$-norm:

\begin{corollary} 
There is an $O (k\log ( {1}/{\eps}))$-query algorithm for $\NonExp(\eps,k)$. 
\end{corollary}

Another corollary of \Cref{theorem: main} is about finding a \emph{strong} $\eps$-fixed point
  in a contraction map $f$.
We say $x$ is a strong $\eps$-fixed point of $f$ if $\norm{x-x^*}\leq \eps$, where $x^*$ is the unique fixed point of $f$.
The following observation leads to \Cref{coro2}, where 
  $\StrongC(\eps,\gamma,k)$ denotes the problem of finding a strong $\eps$-fixed point:

\begin{observation}
Let $f$ be a $(1-\gamma)$-contraction map and $x^*$ be its unique fixed point.
	Let $x$ be any $(\eps\gamma)$-fixed point of $f$, i.e., $x$ satisfies $\norm{f(x)-x}\leq \eps\gamma$. Then we have $$\norm{x-x^*}\leq \norm{x-f(x)}+\norm{f(x)-x^*}\leq \eps\gamma+(1-\gamma)\norm{x-x^*},$$ which implies $\norm{x-x^*}\leq \eps$.
This gives a black-box reduction from $\StrongC(\eps,\gamma,k)$ to $\Contraction(\eps\gamma,\gamma,k)$, which is both query-efficient and time-efficient.
\end{observation}
\begin{corollary}\label{coro2}
	There is an $O (k\log ( {1}/({\eps\gamma})))$-query algorithm for $\StrongC(\eps,\gamma,k)$.
\end{corollary}

In sharp contrast with \Cref{coro2}, however, we show that it is impossible to strongly approximate an exact fixed point in a non-expansive map over $[0,1]^2$ under the $\ell_\infty$ norm.

\begin{restatable}{theorem}{impossiblestrongappxnonexpanding}\label{thm:strong-appx}
There is no deterministic or randomized algorithm which, when given oracle access to any non-expansive map  $f: [0,1]^2 \mapsto [0,1]^2$ under the $\ell_{\infty}$-norm, computes in an expected bounded number of queries a point that is within distance $1/4$ of an exact fixed point of $f$.
\end{restatable} 

The impossibility of strong approximation for expansive maps under the $\ell_\infty$ norm, 
i.e. maps with Lipschitz constant $L>1$,
was shown previously by \cite{SS03}, and it was conjectured that the same fact should hold 
also for non-expansive maps ($L=1$). For the $\ell_2$ norm, impossibility of
strong approximation for non-expansive maps was shown earlier \cite{Si01}.

The problem $\Contraction(\eps,\gamma,k)$ is a \emph{promise problem}, i.e.,
it is promised that the function $f$ in the black-box (the oracle) is a $(1-\gamma)$-contraction.
In the various relevant applications (stochastic games etc.), the corresponding function that is induced is by construction a contraction, thus it is appropriate in these cases to restrict attention to functions that satisfy the contraction promise.

For any promise problem, one can define a corresponding total
search problem, where the black-box can be any function $f$ on the domain and the problem is to compute either a solution or a violation of the promise.
In our case, the corresponding total search problem, denoted
$\TContraction(\eps, \gamma, k)$, is the problem of computing for a given
function $f: [0,1]^k \mapsto [0,1]^k$ either an $\eps$-fixed point
or a violation of the contraction property, i.e. a pair of points $x,y \in [0,1]^k$
such that $\norm{f(x)-f(y)} > (1-\gamma)\norm{x-y}$.
For any promise problem, the corresponding total search problem is clearly at least as hard as the promise problem.
For some problems it can be strictly harder (and it may depend on the type of
violation that is desired). However, we show that in our case the
two versions have the same query complexity.

\begin{restatable}{theorem}{totalversion}\label{thm: total search version}
	There is an $O (k\log ({1}/{\eps}))$-query algorithm for $\TContraction(\eps,\gamma,k)$.
\end{restatable}

Similar results hold for $\NonExp(\eps,k)$ and $\StrongC(\eps, \gamma,k)$:
the total search versions have the same query complexity as the corresponding promise problems.

\vspace{0.2cm}
\noindent\textit{Remark.} It is also important to note that while our algorithm in \Cref{theorem: main} is query-efficient, we do not currently see a way of making it time-efficient. The algorithm guarantees that within polynomial queries we can find an $\eps$-fixed point, but each iteration requires a brute force procedure to \emph{determine the next query point}. 
We will explain more details of techniques in \Cref{section: Sketch of the Query Algorithm}.

\vspace{0.2cm}
\noindent {\bf Other Related Work.}
We have already mentioned the most relevant works addressing the query complexity
of computing the fixed point of a contraction map.
For continuous functions $f: [0,1]^k \mapsto [0,1]^k$ that have Lipschitz constant greater than 1 (i.e. are expansive), there are exponential lower
bounds on the query complexity of computing a (weak) approximate fixed point \cite{HPV89,CD08}. 
In a succinct black-box model, it is known that the query complexity of $\TFNP$ problems is polynomial in the description size of the object in the box \cite{KNY19}. Their results do not apply to our setting. For instance, in the $\TContraction(\eps,\gamma,k)$ problem, the objects are essentially all continuous functions over $[0,1]^k$, whose description size far exceeds an exponential function of $k$.

$\Contraction(\eps,\gamma,k)$ when considered in the white-box model\footnote{The white-box model refers to the model where the function is explicitly given by a polynomial-size circuit, in contrast to the black-box model where the function can only be accessed via an oracle as we studied in this paper. When we talk about a computational problem under the white-box model, we measure the efficiency by the time complexity.} can be formulated as a total search problem so that it lies in the class $\TFNP$. In fact, it is one of the motivating problems in~\cite{DP11} to define the class $\CLS$ for capturing problems that lie in both $\PLS$~\cite{JPY88} and $\PPAD$~\cite{Pap94}. Later, the problem of computing the exact fixed point of Piecewise-Linear contraction maps, an easier variant of $\Contraction(\eps,\gamma,k)$, was placed in $\UEOPL$~\cite{FGMS20}, a subclass of $\CLS$ to capture problems with a unique solution. It is not known that $\Contraction(\eps,\gamma,k)$ is complete for any $\TFNP$ class. Notably, to the best of our knowledge, for the known fixed point problems that are complete for some $\TFNP$ class, their query complexity in the black-box model is exponential. 
Examples of such well-known problems include $\PPAD$-complete problems $\textsc{Brouwer}$ and $\textsc{Sperner}$~\cite{Pap94,CD09}, $\PPA$-complete problems $\textsc{Borsuk–Ulam}$, $\textsc{Tucker}$~\cite{Pap94,ABB2020} and $\textsc{M{\"o}biusSperner}$~\cite{deng2021understanding}, $\CLS$-complete problems $\textsc{KKT}$~\cite{FGHS23} and $\textsc{MetricBanach}$\footnote{$\textsc{MetricBanach}$ refers to the problem of computing an approximate fixed point of a contraction map where the distance function $d$ is also part of the input.}~\cite{DTZ18}, and $\UEOPL$-complete problem $\textsc{OPDC}$~\cite{FGMS20}. 

However, our results indicate that $\Contraction(\eps,\gamma,k)$ is dramatically different from all these fixed point problems above in terms of query complexity. 
Indeed, it is a necessary step towards obtaining a polynomial-time general-purpose (i.e., black-box) algorithm for $\Contraction(\eps,\gamma,k)$, which, ideally if true, would imply many breakthroughs in the fields of verification, semantics, learning theory, and game theory as we discussed before.

\subsection{Sketch of the Main Algorithm}
\label{section: Sketch of the Query Algorithm}

We give a high-level sketch of the main query algorithm for \Cref{theorem: main}.
We start by discretizing the search space. 
Let $g:[0,n]^k \mapsto [0,n]^k$ with $g(x):=n\cdot f(x/n)$ and $n:=\lceil 16/(\gamma\eps)\rceil$.
It is easy to show that $g$ remains a $(1-\gamma)$-contraction over $[0,n]^k$ and it 
  suffices to find a $(16/\gamma)$-fixed point of $g$.
Moreover, by rounding the unique fixed point $x^*$ of $g$ to an integer point,
  we know trivially that there exists at least one point $x$ in $\EVEN(n,k)$\footnote{Looking ahead, the usage of $\EVEN(n,k)$ will guarantee that all queried points lie in the normal grid $\set{0,\cdots,n}^k$, which implies that the approximate fixed point that we will find lies in the normal grid $\set{0,\cdots,n}^k$ as well.}, which is the set of integral points in the grid $\{0,1,\ldots,n\}^k$ with all coordinates being even, that satisfies $\|x-x^*\|_\infty\le 1$. Furthermore, 
  it is easy to show that any such point $x$ must be a $(16/\gamma)$-fixed point.
So our goal is to find a point $x\in \EVEN(n,k)$ that satisfies $\|x-x^*\|_\infty \le 1$ query-efficiently.

To this end, we use $\Cand^t$ to denote the set of $\EVEN(n,k)$ that remains possible to
  be close to the unknown exact fixed point $x^*$ of $g$.
Starting with $\Cand^0$ set to be the full set $\EVEN(n,k)$, the success of the
  algorithm relies on whether we can cut down the size of $\Cand^t$ efficiently.
For this purpose we prove a number of geometric lemmas in \Cref{section: geometry} to 
  give a characterization of the exact fixed point $x^*$, which lead to the following primitive
  used by the algorithm repeatedly:
\begin{flushleft}\begin{quote}
Given $x\in [0,n]^k$, $i\in [k]$ and $\phi\in \{\pm 1\}$, we write $\calP_i(x,\phi)$ to denote
  the set of points $y\in [0,n]^k$ such that $\phi\cdot (y_i-x_i)=\|y-x\|_\infty$, where $\calP$ is a shorthand for pyramid; see \Cref{figure: 11} for illustrations.
Then after querying a point $a\in [0,n]^k$, either $a$ was found to be a $(16/\gamma)$-fixed point
  (in which case the algorithm is trivially done), or one can find $\phi_i\in \{\pm 1\}$ for 
  each $i\in [k]$ such that no point in $\calP_i(a,\phi_i)$ can be close (within $\ell_\infty$-distance $1$) to $x^*$ (in which case we can update $\Cand^t$ by removing all points
  in $\cup_{i\in [k]} \calP_i(a,\phi_i)$). 	
\end{quote}\end{flushleft}

Given this, it suffices to show that for any set of points $T\subseteq \EVEN(n,k)$ (as $\Cand^t$),
  there exists a point $a$ to be queried such that for any $\phi_i\in \{\pm 1\}$:
$$
\left|T\cap\left(\bigcup_{i\in [k]} \calP_i(a,\phi_i)\right)\right|
$$
is large relative to $|T|$. Quantitatively, we show that there exists a point $a$ such that for any $\phi_i\in \{\pm 1\}$,  the set $T\cap\left(\bigcup_{i\in [k]} \calP_i(a,\phi_i)\right)$ has size at least half of $|T|$. This bound is clearly sharp and we refer to such point $a$ as \emph{balanced point}.

To prove the existence of a balanced point,
  we construct an infinite sequence of continuous maps $\{f^t\}$ that 
  can be viewed as relaxed versions of the search for a balanced point.
Using Brouwer's fixed point theorem, every map $f^t$ has a fixed point $p^t$ and 
  thus, by the Bolzano--Weierstrass theorem, there must be an infinite subsequence of $\{p^t\}$ that converges. Letting $p^*\in[0,n]^k$ be the point it converges to, we can show that $p^*$ must be a desired balanced point.  
We could further round $p^*$ to $q^*\in \set{0,\dots,n}^k$ and show that the latter is a balanced point in the grid. 
While we show such a point always exists,
  the brute-force search to find $q^*\in \set{0,\cdots,n}^k$ is the reason why our algorithm is not time-efficient.

\begin{figure}[t]

    \centering%
\begin{subfigure}[b]{0.45\textwidth}
    \centering
    \begin{tikzpicture}[scale=2]
        \coordinate (A) at (0,0);
        \coordinate (B) at (2,0);
        \coordinate (C) at (2,2);
        \coordinate (D) at (0,2);
        \coordinate (Center2D) at (1,1);

        \fill[cyan, opacity=0.5] (A) -- (D) -- (Center2D) -- cycle;
        \draw[thick] (A) -- (B);
        \draw[thick] (B) -- (C);
        \draw[thick] (C) -- (D);
        \draw[thick] (D) -- (A);

        \draw[thick, dotted] (Center2D) -- (A);
        \draw[thick, dotted] (Center2D) -- (B);
        \draw[thick, dotted] (Center2D) -- (C);
        \draw[thick, dotted] (Center2D) -- (D);

        \node[below] at (Center2D) {$x$};
    \end{tikzpicture}
    \caption{The light blue region is given by $\calP_1(x,-1)$ where $x$ is the center point.}
\end{subfigure}\hfill
\begin{subfigure}[b]{0.45\textwidth}
    \centering
    \begin{tikzpicture}[scale=2]
        \coordinate (A) at (0,0);
        \coordinate (B) at (2,0);
        \coordinate (C) at (2,2);
        \coordinate (D) at (0,2);
        \coordinate (E) at (-0.5,0,2);
        \coordinate (F) at (1.5,0,2);
        \coordinate (G) at (1.5,2,2);
        \coordinate (H) at (-0.5,2,2);
        \coordinate (Center3D) at (0.75,1.0,1.0);
        
        \fill[cyan, opacity=0.5] (A) -- (Center3D) -- (D) -- (H) -- (E) -- cycle;

        \draw[thick] (A) -- (B);
        \draw[thick] (B) -- (C);
        \draw[thick] (C) -- (D);
        \draw[thick] (D) -- (A);
        \draw[thick] (E) -- (F);
        \draw[thick] (F) -- (G);
        \draw[thick] (G) -- (H);
        \draw[thick] (H) -- (E);
        \draw[thick] (A) -- (E);
        \draw[thick] (B) -- (F);
        \draw[thick] (C) -- (G);
        \draw[thick] (D) -- (H);

        \draw[thick, dotted] (Center3D) -- (A);
        \draw[thick, dotted] (Center3D) -- (B);
        \draw[thick, dotted] (Center3D) -- (C);
        \draw[thick, dotted] (Center3D) -- (D);
        \draw[thick, dotted] (Center3D) -- (E);
        \draw[thick, dotted] (Center3D) -- (F);
        \draw[thick, dotted] (Center3D) -- (G);
        \draw[thick, dotted] (Center3D) -- (H);

        \node[below=0.3cm] at (Center3D) {$x$};
    \end{tikzpicture}
    \caption{The light blue region is given by $\calP_1(x,-1)$ where $x$ is the center point.}
\end{subfigure}
    \caption{Pyramids in two-dimensional and three-dimensional cubes.}\label{figure: 11}
\end{figure}

\section{Preliminaries}
\begin{definition}[Contraction]
Let $0<\gamma<1$ and $(\calM,d)$ be a metric space. A map $f:\calM\mapsto\calM$ is a $(1-\gamma)$-\emph{contraction map} with respect to $(\calM,d)$ if $d(f(x),f(y))\leq (1-\gamma)\cdot d(x,y)$ for all $x,y\in \calM$.

A map $f:\calM\mapsto\calM$ is said to be \emph{non-expansive} if 
  $d(f(x),f(y))\leq d(x,y)$ for all $x,y\in \calM$.
\end{definition}

Every contraction map has a unique fixed point, i.e., $x^*$ with $f(x^*)=x^*$ by Banach's theorem. Every non-expansive map over $[0,1]^k$ with respect to the infinite norm has a (not necessarily unique) fixed point by Brouwer's theorem.
In this paper, we study the query complexity of finding an $\eps$-fixed point of a $(1-\gamma)$-contraction map $f$  over the $k$-cube $[0,1]^k$ with respect to the infinity norm:

\begin{definition}[$\Contraction(\eps,\gamma,k)$]
In problem $\Contraction(\eps,\gamma,k)$, we are given 
  oracle access to a $(1-\gamma)$-contraction map $f$ over $[0,1]^k$ with respect to the infinity norm, i.e., $f$ satisfies
$$
\|f(x)-f(y)\|_\infty\le (1-\gamma)\cdot \|x-y\|_\infty,\quad \text{for all $x,y\in [0,1]^k$}
$$
and the goal is to find an $\eps$-fixed point of $f$, i.e., a point $x \in [0,1]^k$ such that $\norm{f(x )-x }\leq \eps$.

We also write $\StrongC(\eps,\gamma,k)$ to denote the problem with the same input but
  the goal is to find a \emph{strong $\eps$-fixed point} of $f$, i.e., $x\in [0,1]^k$
  such that $\|x-x^*\|_\infty\le \eps$, where $x^*$ is the unique fixed point of $f$. 
\end{definition}

We define similar problems for non-expansive maps over $[0,1]^k$:

\begin{definition}[$\NonExp(\eps,k)$]
In problem $\NonExp(\eps,k)$, we are given 
  oracle access to a non-expansive map $f :[0,1]^k\rightarrow [0,1]^k$ with respect to the infinity norm, i.e., $f$ satisfies
$$
\|f(x)-f(y)\|_\infty\le \|x-y\|_\infty,\quad \text{for all $x,y\in [0,1]^k$}
$$
and the goal is to find an $\eps$-fixed point of $f$. 

We write $\StrongNonExp(\eps,k)$ to denote the problem with the same input but
  the goal is to find a \emph{strong $\eps$-fixed point} of $f$, namely, a point $x$ such that $\|x-x^*\|_\infty\le \eps$, where $x^*$ is any fixed point of $f$.
\end{definition}

Let $f:[0,1]^k\rightarrow [0,1]^k$ be a $(1-\gamma)$-contraction map.
For convenience, our main algorithm for $\Contraction(\eps,\gamma,k)$   will~work on $g:\cube\mapsto\cube$ 
  with $$n:=\left\lceil \frac{16}{\gamma\eps}\right\rceil
  \quad\text{and}\quad g(x)\coloneqq n\cdot f(x/n),\quad\text{for all $x\in\cube$}.$$ 
We record the following simple property about $g$:
\begin{lemma}\label{basiclemma}
	The map $g$ constructed above is still a $(1-\gamma)$-contraction and finding an $\eps$-fixed point of $f$ reduces to finding a $({16}/{\gamma})$-fixed point of $g$.
\end{lemma}
\begin{proof}
	For any two points $x,y\in\cube$, we have
	\begin{align*}
	   \norm{g(x)-g(y)} & = \norm{n\cdot f(x/n)-n\cdot f(y/n)}
	     \leq n (1-\gamma) \cdot \norm {x/n-y/n} 
	     	     =(1-\gamma)\norm{x-y}.
	\end{align*}
Suppose we found a $( {16}/{\gamma})$-fixed point $a$ of $g$. We show that $x\coloneqq a/n$ is an $\eps$-fixed point of $f$:
	\begin{align*}
	   \norm{f(x)-x} & = \norm{g(a)/n-a/n}=\frac{1}{n}\cdot \norm{g(a)-a} \leq \frac{1}{n}\cdot \frac{16}{\gamma} \leq \eps.
	\end{align*}
This finishes the proof of the lemma.
\end{proof}

\noindent \textbf{Notation.} Given a positive integer $m$, we use $[m]$ to denote $\{1,\ldots,m\}$ and $[0:m]$ to denote $\set{0,\ldots,m}$. For a real number $t\in \mathbb{R}$, we let $\sgn(t)=1$ if $t>0$, $\sgn(t)=-1$ if $t<0$, and $\sgn(t)=0$ if $t=0$. 
Given positive integers $n$ and $k$, we use $\EVEN(n,k)$ to denote the set of all integer points $x\in \cube$ such that $x_i$ is even for all $i\in[k]$.

For a point $x\in \mathbb{R}^k$ (not necessarily in $[0,n]^k$), a coordinate $i\in[k]$, and a sign $\phi\in\set{\pm 1}$, we use $\calP_i(x,\phi)$ to denote 
$$\calP_i(x,\phi):=\set{y\in\cube: \phi\cdot(y_i-x_i)= \norm{y-x}},$$ where $\calP$ is a shorthand for pyramid.

Given a $(1-\gamma)$-contraction map $g$ over $\cube$, we use $\Fix(g)$ to denote the unique fixed point of $g$. For any point $x\in\cube$, we use $\Around(x)$ to denote the set 
$$\Around(x):=\big\{y\in[0,n]^k: \norm{x-y}\leq 1\big\}.$$ We note that $\norm{g(y)-y}\leq 2$ for all $y\in\Around(\Fix(g))$ and thus, any point in $\Around(\Fix(g))$ is a desired $({16}/{\gamma})$-fixed point (given that $\gamma<1$). To see this, letting $x=\Fix(g)$, we have 
$$\norm{g(y)-y}\leq \norm{g(y)-g(x)}+\norm{x-y}\leq (2-\gamma)\norm{x-y}\leq 2,
\quad\text{for any $y\in \Around(x)$}.$$ Note that for any $x\in[0,n]^k$, $\Around(x)\cap \EVEN(n,k)$ is non-empty.

\section{Characterizing the Unique Fixed Point}
\label{section: geometry}
In this section, we prove three geometrical lemmas for characterizing the exact fixed point $\Fix(g)$. The corresponding geometrical illustrations are in \Cref{figure: lemma 2}, \Cref{figure: lemma 3}, and \Cref{figure: lemma 4}.

\begin{lemma}\label{lemma: region of fixed point}
	Let $g:\cube\mapsto\cube$ be a $(1-\gamma)$-contraction map and let $a\in\cube$ be a point such that $\norm{g(a)-a}> {16}/{\gamma}$  and $s\in\set{\pm 1,0}^k$ be the sign vector such that $s_i=\sgn(g(a)_i-a_i)$. Then $$\emph{\Fix}(g)\in \bigcup_{i:s_i\neq 0}\calP_i(a+4s,s_i).$$
\end{lemma}

\begin{figure}

\begin{subfigure}[b]{0.45\textwidth}\centering

    \begin{tikzpicture}[scale=0.5]
        \coordinate (A) at (0,0);
        \coordinate (B) at (10,0);
        \coordinate (C) at (10,10);
        \coordinate (D) at (0,10);
        \coordinate (E) at (0,8);
        \coordinate (F) at (8,0);
        \coordinate (G) at (1,0);
        \coordinate (H) at (10,9);
        \coordinate (I) at (0,10);
        \coordinate (J) at (10,0);
        \coordinate (a) at (4.5,3.5);
        \coordinate (ga) at (5.5,7);
        \coordinate (aplus4) at (5.5,4.5);

        \fill[cyan, opacity=0.5] (I) -- (D) -- (C) -- (B) -- (J) -- cycle;
        \draw[thick] (A) -- (B);
        \draw[thick] (B) -- (C);
        \draw[thick] (C) -- (D);
        \draw[thick] (D) -- (A);

        \draw[thick, dotted] (G) -- (H);
        \draw[thick, dotted] (E) -- (F);

        \node[below] at (a) {$a$};
        \node[above] at (ga) {$g(a)$};
        \node[right] at (aplus4) {$a+4s$};
        \fill[black] (a) circle (0.1);
        \fill[black] (ga) circle (0.1);
        \fill[black] (aplus4) circle (0.1);
    \end{tikzpicture}

    \caption{$a$ and $g(a)$ satisfy $s=(+1,+1)$.}
\end{subfigure}\hfill
\begin{subfigure}[b]{0.45\textwidth}\centering%

    \begin{tikzpicture}[scale=0.5]
        \coordinate (A) at (0,0);
        \coordinate (B) at (10,0);
        \coordinate (C) at (10,10);
        \coordinate (D) at (0,10);
        \coordinate (E) at (0,8);
        \coordinate (F) at (8,0);
        \coordinate (G) at (1,0);
        \coordinate (H) at (10,9);
        \coordinate (I) at (0,9);
        \coordinate (J) at (10,10);
        \coordinate (a) at (4.5,3.5);
        \coordinate (ga) at (4.5,7);
        \coordinate (aplus4) at (4.5,4.5);

        \fill[cyan, opacity=0.5] (D) -- (I) -- (aplus4) -- (J) -- (C) -- cycle;
        \draw[thick] (A) -- (B);
        \draw[thick] (B) -- (C);
        \draw[thick] (C) -- (D);
        \draw[thick] (D) -- (A);

        \draw[thick, dotted] (G) -- (H);
        \draw[thick, dotted] (E) -- (F);

        \node[below] at (a) {$a$};
        \node[above] at (ga) {$g(a)$};
        \node[right] at (aplus4) {$a+4s$};
        \fill[black] (a) circle (0.1);
        \fill[black] (ga) circle (0.1);
        \fill[black] (aplus4) circle (0.1);
    \end{tikzpicture}

    \caption{$a$ and $g(a)$ satisfy $s=(0,+1)$. }
\end{subfigure}

\caption{An illustration of \Cref{lemma: region of fixed point}. The light blue region is given by $\bigcup_{i:s_i\neq 0}\calP_i(a+4s,s_i)$, which contains $\Fix(g)$ (the unique fixed point of $g$).}
\label{figure: lemma 2}
\end{figure}

\begin{proof}
Let $c=a+4s$ and $x^*=\Fix(g)$ be the unique fixed point.	
	
	First, we show that  
	$\norm{x^*-a}> {8}/{\gamma}$. Otherwise, we have 
	$$\norm{g(a)-a}\leq \norm{g(a)-g(x^*)}+\norm{x^*-a}\leq (1-\gamma+1)\cdot \norm{x^*-a}\leq \frac{16}{\gamma},$$ which contradicts the assumption that $\norm{g(a)-a}> {16}/{\gamma}$.
	Now it suffices to show that
	  $g(x)\ne x$ for any point $x\in [0,n]^k$ that satisfies both 
	$$x\notin \bigcup_{i,s_i\neq 0}\calP_i(c,s_i)\quad\text{and}\quad\norm{x-a}>\frac{8}{\gamma}.$$ Let $j$ be a dominating coordinate between $x$ and $c$, i.e., a $j\in [k]$ such that $|x_j-c_j|=\norm{x-c}$. We divide the proof into two parts.\medskip 
	
	\noindent \textbf{Part 1: $s_j=0$.} Thus $g(a)_j=a_j$ and $c_j=a_j$. Assume without loss of generality that $x_j\geq c_j$;~the case when $x_j\leq c_j$ is symmetric. On the one hand, we have $x_j-c_j=\norm{x-c}$, which gives	\begin{equation}
		x_j=c_j+\norm{x-c}. \tag{$\star$}
	\end{equation}
	 On the other hand, using that $g$ is a $(1-\gamma)$-contraction and 
	   $g(a)_j=a_j$, we have  
$$g(x)_j-a_j=g(x)_j-g(a)_j\leq (1-\gamma)\cdot \norm{x-a}.$$ Equivalently,  $g(x)_j\leq a_j+\norm{x-a}-\gamma\norm{x-a}$. Combining with $\norm{x-a}> {8}/{\gamma}$, this implies 
	 \begin{equation}
	 	g(x)_j<a_j+\norm{x-a}-8.\tag{$\diamond$}
	 \end{equation}
Putting $(\star)$ and $(\diamond)$ together and the facts that $\norm{c-a}=4$ and $c_j=a_j$, we have  
	\begin{align*}
	   x_j-g(x)_j & > \norm{x-c}-\norm{x-a}+8 \geq -\norm{c-a}+8  > 0,
	\end{align*}
	which implies that $g(x)_j\neq x_j$ and $g(x)\neq x$.\medskip
	 	
	\noindent\textbf{Part 2: $s_j\neq 0$.} Assume without loss of generality that $s_j=+1$; the case $s_j=-1$ is symmetric. 
	
	Since we are considering points not in $\smash{\bigcup_{i:s_i\neq 0}\calP_i(c,s_i)}$, it must be the case that $x\in \calP_j(c,-s_j)$ and thus, $x_j\leq c_j$. 
	Since $\smash{\norm{x-a}> {8}/{\gamma}}$, we have $ x_j\leq a_j$; otherwise $ a_j\leq x_j\leq c_j$ and $\norm{x-c}\leq c_j-x_j\leq c_j-a_j\leq 4$ and thus,
$$\norm{x-a}\le \norm{x-c}+\norm{a-c}\le 8.$$
	
	Given the $(1-\gamma)$-contraction of $g$, we have $g(a)_j-g(x)_j\leq (1-\gamma)\cdot\norm{x-a}$, which implies 
	\begin{equation}\label{eq: hahahaha}
		 g(x)_j\geq  g(a)_j-\norm{x-a}+\gamma\norm{x-a}> g(a)_j-\norm{x-a}+8.
	\end{equation}
	Next, we show an upper bound on $\norm{x-a}$. 
	Recall that $x\in\calP_j(c,-s_j)$. Consider $y=x-4s$. We have $y\in\calP_j(a,-s_j)$. So $$\norm{x-a}\leq\norm{x-y}+\norm{y-a}=4+ (a_j-y_j)=4+ (a_j-x_j+4 )=8+ (a_j-x_j).$$

	Given this and plugging the upper bound in \Cref{eq: hahahaha}, we will get $$ g(x)_j> g(a)_j+x_j-a_j .$$ Recall that $s_j=+1$ implies that $g(a)_j>a_j$. So we have $ g(x)_j> x_j$ and $g(x)\neq x$.	
This finishes the proof of \Cref{lemma: region of fixed point}.
\end{proof}

\begin{lemma}\label{lemma: around}
Let $b\in\mathbb{R}^k$ and $s\in\set{\pm 1,0}^k$ such that $s\neq 0^k$.
Then every $x\in \bigcup_{i:s_i\neq 0}\calP_i(b+2s,s_i)$ must have $\emph{\Around}(x)\subseteq \bigcup_{i:s_i\neq 0}\calP_i(b,s_i)$.
\end{lemma}
\begin{figure}
\begin{subfigure}[b]{0.45\textwidth}\centering

    \begin{tikzpicture}[scale=0.5]
        \coordinate (A) at (0,0);
        \coordinate (B) at (10,0);
        \coordinate (C) at (10,10);
        \coordinate (D) at (0,10);
        \coordinate (E) at (0,8);
        \coordinate (F) at (8,0);
        \coordinate (G) at (1,0);
        \coordinate (H) at (10,9);
        \coordinate (I) at (0,9);
        \coordinate (J) at (9,0);
        \coordinate (K) at (0,10);
        \coordinate (L) at (10,0);
        \coordinate (a) at (4.5,3.5);
        \coordinate (b) at (5,4);
        \coordinate (bplus2) at (5.5,4.5);
        \coordinate (xstar) at (3.5,6.7);
        \coordinate (xa) at (3.5-0.25,6.7-0.25);
        \coordinate (xb) at (3.5-0.25,6.7+0.25);
        \coordinate (xc) at (3.5+0.25,6.7+0.25);
        \coordinate (xd) at (3.5+0.25,6.7-0.25);

        \fill[cyan, opacity=0.5] (K) -- (D) -- (C) -- (B) -- (L) -- cycle;
        \draw[thick] (A) -- (B);
        \draw[thick] (B) -- (C);
        \draw[thick] (C) -- (D);
        \draw[thick] (D) -- (A);
        
        \draw[thick] (xa) -- (xb);
        \draw[thick] (xb) -- (xc);
        \draw[thick] (xc) -- (xd);
        \draw[thick] (xd) -- (xa);

        \draw[thick, dotted] (G) -- (H);
        \draw[thick, dotted] (E) -- (F);
        \draw[thick, dotted] (I) -- (J);

        \node[below] at (a) {$a$};
        \fill[black] (a) circle (0.1);
        
		\node at ([xshift=0.3cm, yshift=-0.3cm] b) {$b$};
        \fill[black] (b) circle (0.1);
        
        \node[right] at (bplus2) {$b+2s$};
        \fill[black] (bplus2) circle (0.1);
        
        \node[right] at (xstar) {$x$};
        \fill[black] (xstar) circle (0.1);
        
    \end{tikzpicture}

    \caption{The case $s=(+1,+1)$.}
\end{subfigure}\hfill
\begin{subfigure}[b]{0.45\textwidth}\centering%

    \begin{tikzpicture}[scale=0.5]
        \coordinate (A) at (0,0);
        \coordinate (B) at (10,0);
        \coordinate (C) at (10,10);
        \coordinate (D) at (0,10);
        \coordinate (E) at (0,8);
        \coordinate (F) at (8,0);
        \coordinate (G) at (1,0);
        \coordinate (H) at (10,9);
        \coordinate (I) at (0,9);
        \coordinate (J) at (10,10);
        \coordinate (K) at (0,8.5);
        \coordinate (L) at (10,9.5);
        \coordinate (a) at (4.5,3.5);
        \coordinate (b) at (4.5,4);
        \coordinate (bplus2) at (4.5,4.5);
        \coordinate (x) at (2.5,6.7);
        \coordinate (xa) at (2.5-0.25,6.7-0.25);
        \coordinate (xb) at (2.5-0.25,6.7+0.25);
        \coordinate (xc) at (2.5+0.25,6.7+0.25);
        \coordinate (xd) at (2.5+0.25,6.7-0.25);

        \fill[cyan, opacity=0.5] (D) -- (I) -- (bplus2) -- (J) -- (C) -- cycle;
        \draw[thick] (A) -- (B);
        \draw[thick] (B) -- (C);
        \draw[thick] (C) -- (D);
        \draw[thick] (D) -- (A);
        
        \draw[thick] (xa) -- (xb);
        \draw[thick] (xb) -- (xc);
        \draw[thick] (xc) -- (xd);
        \draw[thick] (xd) -- (xa);

        \draw[thick, dotted] (G) -- (H);
        \draw[thick, dotted] (E) -- (F);
        \draw[thick, dotted] (b) -- (K);
        \draw[thick, dotted] (b) -- (L);

        \node[below] at (a) {$a$};
        \fill[black] (a) circle (0.1);
        \node[right] at (b) {$b$};
        \fill[black] (b) circle (0.1);
        \node[right] at (bplus2) {$b+2s$};
        \fill[black] (bplus2) circle (0.1);
        \node[right] at (x) {$x$};
        \fill[black] (x) circle (0.1);
    \end{tikzpicture}

    \caption{The case $s=(0,+1)$. }
\end{subfigure}

\caption{An illustration of \Cref{lemma: around}. The light blue region is given by $\bigcup_{i:s_i\neq 0}\calP_i(b+2s,s_i)$, which contains $x$. $\Around(x)$ is the small square around $x$, which is contained in $\bigcup_{i:s_i\neq 0}\calP_i(b,s_i)$.}
\label{figure: lemma 3}
\end{figure}

\begin{proof}
Let $i^*$ be such that $s_{i^*}\neq 0$ and $x\in\calP_{i^*}(b+2s,s_{i^*})$. 
Assume without loss of generality that $s_{i^*}=+1$.
We have  $x_{i^*}-(b_{i^*}+2 ) \geq |x_i-(b_i+2s_i)|$ for every $i\in [k]$. 

Fix an arbitrary $y\in \Around(x)$. 
We must have $y_{i^*}-b_{i^*}\ge 0$, which follows from 
$$0\le x_{i^*}-(b_{i^*}+2) \le y_{i^*}+1-(b_{i^*}+2)$$
and thus, $y_{i^*}\ge b_{i^*}+1$. 
Let $$j\in\argmax_{i\in[k]}\big\{s_i(y_i-b_i)\mid s_i\neq 0\big\}.$$
Since $s_{i^*}=+1$, we have $$s_j(y_j-b_j)\geq  y_{i^*}-b_{i^*}\ge 0,$$ thus $s_j(y_j-b_j)=|y_j-b_j|$.
Our goal is to show $y\in\calP_j(b,s_j)$ and it suffices for us to show $|y_j-b_j|\geq |y_i-b_i|$ for all $i\in[k]$.

Let's consider first an arbitrary $i\in [k]$ with $s_i=0$. Recall that $ x_{i^*}-(b_{i^*}+2 ) \geq |x_i- b_i |$. In particular, this implies $|x_{i^*}-b_{i^*}|-2\geq |x_i-b_i|$. Putting everything together, we have
$$|y_j-b_j|\geq |y_{i^*}-b_{i^*}|\geq |x_{i^*}-b_{i^*}|-1\geq |x_i-b_i|+1\geq |y_i-b_i|.$$

Finally let's consider an $i\in [k]$ such that $s_i\neq 0$. If $s_i(y_i-b_i)>0$, by the definition of how we picked $j$, we have $|y_j-b_j| \geq |y_i-b_i|$. If $s_i(y_i-b_i)<0$, then we have
	\begin{align*}
|y_i-b_i| & =|y_i-(b_i+2s_i)|-2 \leq |x_i-(b_i+2s_i)|-1\leq  x_{i^*}- b_{i^*}-3 \le y_{i^*}-b_{i^*}-2	
	  < |y_j-b_j|.	\end{align*}
	This finishes the proof of the lemma.
\end{proof}

\begin{lemma}\label{lemma: eliminate}
Let $a\in\cube$ and $s\in\set{\pm 1,0}^k$ such that $s\neq 0^k$. Then for every $j\in[k]$, there exists $\phi\in\set{\pm 1}$ such that $$\calP_j(a,\phi)\cap \left(\bigcup_{i:s_i\neq 0}\calP_i(a+2s,s_i)\right)=\emptyset.$$
\end{lemma}

\begin{figure}
\begin{subfigure}[b]{0.45\textwidth}\centering

    \begin{tikzpicture}[scale=0.5]
        \coordinate (A) at (0,0);
        \coordinate (B) at (10,0);
        \coordinate (C) at (10,10);
        \coordinate (D) at (0,10);
        \coordinate (E) at (0,8);
        \coordinate (F) at (8,0);
        \coordinate (G) at (1,0);
        \coordinate (H) at (10,9);
        \coordinate (I) at (0,9);
        \coordinate (J) at (9,0);
        \coordinate (a) at (4.5,3.5);
        \coordinate (aplus2) at (5,4);

        \fill[cyan, opacity=0.5] (I) -- (D) -- (C) -- (B) -- (J) -- cycle;
        \draw[thick] (A) -- (B);
        \draw[thick] (B) -- (C);
        \draw[thick] (C) -- (D);
        \draw[thick] (D) -- (A);

        \draw[thick, dotted] (G) -- (H);
        \draw[thick, dotted] (E) -- (F);

        \node[below] at (a) {$a$};
        \node[right] at (aplus2) {$a+2s$};
        \fill[black] (a) circle (0.1);
        \fill[black] (aplus2) circle (0.1);
    \end{tikzpicture}

    \caption{The case $s=(+1,+1)$. The light blue region is given by $\bigcup_{i:s_i\neq 0}\calP_i(a+2s,s_i)$. Clearly it does not intersect with $\calP_1(a,-1)$ or $\calP_2(a,-1)$.}
\end{subfigure}\hfill
\begin{subfigure}[b]{0.45\textwidth}\centering%

    \begin{tikzpicture}[scale=0.5]
        \coordinate (A) at (0,0);
        \coordinate (B) at (10,0);
        \coordinate (C) at (10,10);
        \coordinate (D) at (0,10);
        \coordinate (E) at (0,8);
        \coordinate (F) at (8,0);
        \coordinate (G) at (1,0);
        \coordinate (H) at (10,9);
        \coordinate (I) at (0,8.5);
        \coordinate (J) at (10,9.5);
        \coordinate (a) at (4.5,3.5);
        \coordinate (aplus2) at (4.5,4);

        \fill[cyan, opacity=0.5] (D) -- (I) -- (aplus2) -- (J) -- (C) -- cycle;
        \draw[thick] (A) -- (B);
        \draw[thick] (B) -- (C);
        \draw[thick] (C) -- (D);
        \draw[thick] (D) -- (A);

        \draw[thick, dotted] (G) -- (H);
        \draw[thick, dotted] (E) -- (F);

        \node[below] at (a) {$a$};
        \node[right] at (aplus2) {$a+2s$};
        \fill[black] (a) circle (0.1);
        \fill[black] (aplus2) circle (0.1);
    \end{tikzpicture}

    \caption{The case $s=(0,+1)$. The light blue region is given by $\bigcup_{i:s_i\neq 0}\calP_i(a+2s,s_i)$. Clearly it does not intersect with $\calP_2(a,-1)$, $\calP_1(a,-1)$, or $\calP_1(a,1)$.}
\end{subfigure}

\caption{An illustration of \Cref{lemma: eliminate}. When $s_i\neq 0$, we can prove that $\bigcup_{i:s_i\neq 0}\calP_i(a+2s,s_i)$ does not intersect with $\calP_i(a,-s_i)$; see left figure above. When $s_i=0$, we can prove $\bigcup_{i:s_i\neq 0}\calP_i(a+2s,s_i)$ does not intersect with $\calP_i(a,-1)$ or $\calP_i(a,1)$; see right figure above.}
\label{figure: lemma 4}
\end{figure}

\begin{proof}
First we note that $$\left(\bigcup_{i:s_i\neq 0}\calP_i(a,-s_i)\right)\cap \left(\bigcup_{i:s_i\neq 0}\calP_i(a+2s,s_i)\right)=\emptyset.$$ This implies that $$\calP_j(a,\phi)\cap \left(\bigcup_{i:s_i\neq 0}\calP_i(a+2s,s_i)\right)=\emptyset$$ for all $j$ with $s_j\neq 0$ by setting $\phi=-s_j$. 

	Now consider a $j$ with $s_j=0$. Under this case, we show in fact that $$\calP_j(a,\phi)\cap \left(\bigcup_{i:s_i\neq 0}\calP_i(a+2s,s_i)\right)=\emptyset$$ for both $\phi\in\set{\pm 1}.$ 
	Consider any point $x\in \bigcup_{i:s_i\neq 0}\calP_i(a+2s,s_i)$ and we show that $x\notin \calP_j(a,-1)$ and  $x\notin \calP_j(a,+1)$. 
	Let $b=a+2s$. As $x\in \bigcup_{i:s_i\neq 0}\calP_i(b,s_i)$, there exists $i^*$ with $s_{i^*}\neq 0$ such that $$s_{i^*}(x_{i^*}-b_{i^*})=\norm{x-b}\geq |x_j-b_j|.$$ Note also that $|x_{i^*}-a_{i^*}|=s_{i^*}(x_{i^*}-b_{i^*}+2s_{i^*})=s_{i^*}(x_{i^*}-b_{i^*})+2$ and $b_j=a_j$, so we have $$|x_{i^*}-a_{i^*}|=\norm{x-b}+2>|x_j-b_j|=|x_j-a_j|.$$ Thus  $x\notin \calP_j(a,-1)$ and $x\notin \calP_j(a,+1)$.		
	This finishes the proof of the lemma.
\end{proof}

\section{The Algorithm}\label{sec:mainalg}

\begin{algorithm}[!t]
	\caption{Query Algorithm for $\Contraction(\eps,\gamma,k)$} 
	\label{alg: main}
\begin{algorithmic}[1]	
 \State
	Let $\Cand^0 \gets \EVEN(n,k)$.  
 \For {$t=1,2,\ldots$}
	\State	Find and \textbf{query} an $a^t\in[0:n]^k$ such that $a^t$ is a balanced point of $\Cand^{t-1}$.
		 \State \textbf{if} {$\norm{g(a^t)-a^t}\leq  {16}/{\gamma}$} \textbf{then}
		 \textbf{return} $a^t$ as a $(16/\gamma)$-fixed point of $g$\label{line: return}
		 
		\State Let $\smash{s\in\{\pm 1,0\}^k}$ be such that $\smash{s_i=\sgn(g(a^t)_i-a^t_i)}$ for all $i\in[k]$,  $b^t\gets a^t+2s$, and
		 \begin{equation}\label{Candt}\Cand^t\gets \Cand^{t-1}\cap \left(\bigcup_{i:s_i\neq 0}\calP_i(b^t,s_i)\right).\end{equation}
 \EndFor
\end{algorithmic}
\end{algorithm}

We prove \Cref{theorem: main} in this section.
Our algorithm for $\Contraction(\eps,\gamma,k)$ is described in \Cref{alg: main}.
Given oracle access to a $(1-\gamma)$-contraction map 
  $g:[0,n]^k\rightarrow [0,n]^k$, we show that it can find a 
  $( {16}/{\gamma})$-fixed point of $g$  
within $O\left(k\log\left(n\right)\right)$ many queries. This is sufficient given \Cref{basiclemma}. 

The analysis of \Cref{alg: main} uses the following theorem which 
  we prove in the next section. 
In particular, it guarantees the existence of 
  the point $a^t$ to be queried in round $t$ that satisfies that for any $s\in\set{\pm 1}^k$, we have $\left(\bigcup_{i\in[k]}\calP_i(a^t,s_i)\right)$ intersects with at least half of $\Cand^{t-1}$.
We will call a point $q^*$ with the property stated in \Cref{lemma: existence of balanced point}
 below a \emph{balanced} point for $T$.

\begin{restatable}{theorem}{existenceofbalancedpoint}\label{lemma: existence of balanced point}
For any $T\subseteq \emph{\EVEN}(n,k)$, there exist a point $q^*\in[0:n]^k$ such that $$\left|T\cap \left(\bigcup_{i\in[k]}\calP_i(q^*,s_i)\right)\right|\geq \frac{1}{2}\cdot |T|,\quad\text{for all $s\in\set{\pm 1}^k$}.$$
\end{restatable}

At a high level, \Cref{alg: main} maintains a subset of grid points $\EVEN(n,k)$ as candidate solutions, which is denoted by $\Cand^t$ after round $t$. We show the following invariants:
	
	\begin{lemma}\label{lemma: for alg}
		For every round $t\geq 1$, either the point $a^t$ queried is a $(16/\gamma)$-fixed point of $g$ (and the algorithm terminates), or we have both 
		 $(\emph{\Around}(\emph{\Fix}(g))\cap \emph{\EVEN}(n,k))\subseteq\emph{\Cand}^t$ and
		 \begin{equation}\big|\emph{\Cand}^t\big|\leq \frac{1}{2}\cdot \big|\emph{\Cand}^{t-1}\big|.\label{easypart}\end{equation}
	\end{lemma}
	
	\begin{proof}[Proof of \Cref{lemma: for alg}]
We start with the proof of (\ref{easypart}).
Suppose that $a^t$ is a balanced point of $\Cand^{t-1}$. 
Then by \Cref{lemma: eliminate}, for any coordinate $j\in[k]$, there must exist a sign $\phi_j\in\set{\pm 1}$ such that $$\calP_j(a^t,\phi_j)\cap \left(\bigcup_{i:s_i\neq 0}\calP_i(b^t,s_i)\right)=\emptyset.$$ Since $\Cand^t\coloneqq \Cand^{t-1}\cap \left(\bigcup_{i:s_i\neq 0}\calP_i(b^t,s_i)\right)$, we know that for every coordinate $j\in[k]$, there exists $\phi_j\in\set{\pm 1}$ such that $\Cand^t\cap \calP_j(a^t,\phi_j)=\emptyset$. The inequality (\ref{easypart}) follows directly from the definition of the balanced point.
	
Next we prove by induction that 
  $ ({\Around}( {\Fix}(g))\cap \EVEN(n,k))\subseteq  {\Cand}^t$ for every $t$ before the round that the algorithm terminates.
The basis is trivial given that $\Cand^0$ is set to be $\EVEN(n,k)$.	
For round $t\ge 1$, we assume that $\norm{g(a^t)-a^t}> {16}/{\gamma}$; otherwise a solution is found and the algorithm terminates.
	
	Let $b^t=a^t+2s$ and $c^t=b^t+2s$. By \Cref{lemma: region of fixed point}, we know that $$\Fix(g)\in \left(\bigcup_{i:s_i\neq 0}\calP_i(c^t,s_i)\right).$$ 
	Since $b^t$ is defined as $c^t-2s$, by \Cref{lemma: around}, we know that $$\Around(\Fix(g))\subseteq \bigcup_{i:s_i\neq 0}\calP_i(b^t,s_i).$$ 
We finish the proof by using the inductive hypothesis  $(\Around(\Fix(g))\cap \EVEN(n,k))\subseteq \Cand^{t-1}$ and (\ref{Candt}).	
\end{proof}
		
Recall that ${\Around}( {\Fix}(g))\cap \EVEN(n,k)$ is non-empty. Given \Cref{lemma: for alg} and that $|\Cand^0|\le n^k$, within at most 
$$
 O\left( \log(n^k)\right)=O\left(k\log\left(\frac{1}{\eps\gamma}\right)\right)
$$
many rounds, one of the points $a^t$ queried by the algorithm must be
  a $(16/\gamma)$-fixed point of $g$. 
This finishes the proof of \Cref{theorem: main}.

\section{Existence of Balanced Point}

We prove \Cref{lemma: existence of balanced point} in this section, which we restate below:
\existenceofbalancedpoint*
\begin{proof} 
For each positive integer $t\geq 4$ we let 
$$
S^t:= \cup_{x\in T} B(x,1/t)\subset [-1/4,n+1/4]^k,
$$
where $B(x,1/t)$ denotes the $\ell_2$-ball of radius $1/t$ centered at $x$. We write $\vol(S^t)$ to denote 
the volume of $S^t$ and $\vol(S^t\cap \calP)$ to denote the volume of the intersection of $S^t$ and some pyramid $\calP$.

We apply Brouwer's fixed point to prove the existence of 
  a balanced (real) point for the balls:

\begin{lemma}\label{lem:fixedpoint}
	For every integer $t\geq 4$, there exist $p^*\in [-1/4,n+1/4]^k$ such that $\emph{\vol}(\calP_{i}(p^*,+1)\cap S^t)= \emph{\vol}(\calP_{i}(p^*,-1)\cap S^t)$ holds for all coordinates $i\in[k]$.
\end{lemma}
\begin{proof}
We define a continuous map $f:\left[-1/4 ,n+1/4\right]^k\mapsto \left[-1/4,n+1/4\right]^k$ and apply
Brouwer's fixed point theorem on $f$ to find a fixed point $p^*$ of $f$, and show that
$p^*$ satisfies the property above.  

We define $f$ as follows: For every $p\in \left[-1/4,n+1/4\right]^k$ and $i\in [k]$, let
$$
f_i(p) := p_i + \frac{\vol\big(\calP_i(p,+1)\cap S^t\big)-\vol\big(\calP_i(p,-1)\cap S^t\big)}{(n+0.5)^{k-1}}.
$$
It is clear that $f$ is continuous. To see that it is from $\left[-1/4,n+1/4\right]^k$ to itself,
  we note that 
$$
0\le \frac{\vol\big(\calP_i(p,+1)\cap S^t\big)}{(n+0.5)^{k-1}}\le n+\frac{1}{4}-p_i\quad\text{and}\quad
0\le \frac{\vol\big(\calP_i(p,-1)\cap S^t\big)}{(n+0.5)^{k-1}}\le p_i+\frac{1}{4},
$$
where the upper bounds above simply follow from the fact that $$
\vol\big(\calP_i(p,+1)\big)\leq \bigg(n+\frac{1}{4}-p_i\bigg)\cdot (n+0.5)^{k-1}\quad\text{and}\quad
\vol\big(\calP_i(p,-1)\big)\le \bigg(p_i+\frac{1}{4}\bigg)\cdot(n+0.5)^{k-1}.
$$

As a result, one can apply Brouwer's fixed point theorem on $f$ to conclude that
  there exists a point $p^*\in \left[-1/4,n+1/4\right]^k$ such that $f(p^*)=p^*$, which implies that 
$$
\vol\big(\calP_i(p^*,+1)\cap S^t\big)=\vol\big(\calP_i(p^*,-1)\cap S^t\big)
$$
for all $i\in [k]$.
This finishes the proof of the lemma.
\end{proof}

By \Cref{lem:fixedpoint}, we have that for every $t$,
  there exists $p^t\in \left[-1/4,n+1/4\right]^k$ such that 
$$
\vol\big(\calP_{i}(p^t,+1)\cap S^t\big)=\vol\big(\calP_{i}(p^t,-1)\cap S^t\big)
$$
for all $i\in[k]$.
Given that $\left[-1/4,n+1/4\right]^k$ is compact, $\{p^{t}\}_{t\geq 1}$ has an infinite subsequence that
  converges~to a point $p^*\in \left[-1/4,n+1/4\right]^k$.
We refer to the subsequence as $\{p^{t_\ell}\}_{\ell\ge 1}$.

In \Cref{lemma: balanced 111}, we show that $p^*$ is already a desired balanced point but it may be off-grid.
In \Cref{lemma: balanced 222}, we show how to round $p^*$ to $q^*\in [0:n]^k$ while making sure that 
$$
\calP_{i}(p^*,+1)\cap T\subseteq \calP_{i^*}(q^*,+1)\cap T\quad\text{and}\quad \calP_{i^*}(p^*,-1)\cap T\subseteq 
\calP_{i^*}(q^*,-1)\cap T.
$$
Our goal then follows by combining these two lemmas.

\begin{lemma}\label{lemma: balanced 111}
We have 
$$
\left|\left(\bigcup_{i\in[k]}\calP_i(p^*,s_i)\right)\cap T\right|\geq \frac{1}{2}\cdot |T|,\quad\text{for all $s\in\set{\pm 1}^k$}.
$$
\end{lemma}
\begin{proof}
We write $A$ to denote the following (potentially empty) set of positive real numbers defined using $p^*$:
  $a\in (0,1)$ is in $A$ if there are $i\ne j\in [k]$ such that either 
\begin{enumerate}
\item $p_i^*+p_j^*$ is an integer plus $a$; or
\item $p_i^*+p_j^*$ is an integer minus $a$; or
\item $p_i^*-p_j^*$ is an integer plus $a$; or
\item $p_i^*-p_j^*$ is an integer minus $a$.
\end{enumerate}
Note that it is clear that $A$ is a finite set.

Consider the easier case when $A$ is empty, i.e. $p_i^*+p_j^*$ and $p_i^*-p_j^*$ are
  integers for all $i\ne j\in [k]$. 
Let $\ell$ be a sufficiently large integer such that 
  $1/t_\ell\le 0.1$ and $\|p^{t_\ell}-p^*\|_\infty\le 0.1$. Fixing any $s\in\set{\pm 1}^k$, we first show the following claim:
  \begin{claim}\label{claim: keykeykey}
  	For every point $x\in T$, we have 
  	$$
  	\left(\bigcup_{i\in[k]} \calP_{i}(p^{t_\ell},s_i)\right)\cap B(x,1/t_\ell)\ne \emptyset\ \ \Longrightarrow\ \ 
x\in \bigcup_{i\in[k]} \calP_{i}(p^*,s_i).
  	$$
  \end{claim}
  \begin{proof}
  We show a slightly stronger statement in the proof: For any $i\in[k]$, we have
  $$
  	\calP_{i}(p^{t_\ell},s_i)\cap B(x,1/t_\ell)\ne \emptyset\ \ \Longrightarrow\ \ 
x\in \calP_{i}(p^*,s_i).
  	$$

  Let's prove the contrapositive so take any $x\in T$ such that $x\notin \calP_{i}(p^*,s_i)$. Assume without loss of generality that $s_i=+1$. We know that there exists a $j\ne i$ such that either
$$
x_{i}-p^*_{i}<x_{j}-p^*_j\quad\text{or}\quad 
x_{i}-p^*_{i}<p^*_j-x_j.
$$
For the first case we have $x_{i}-x_j<p^*_{i}-p^*_j$. Since both sides are integers
  we have
\begin{equation}\label{eq:hehe1}
x_{i}-x_j\le p^*_{i}-p^*_j-1
\end{equation}
so intuitively $x$ is far from $\calP_{i}(p^*,+1)$.
From this we can conclude that $B(x,1/t_\ell)\cap \calP_{i}(p^{t_\ell},+1)$ is empty.
To see this is the case, for any  $y\in B(x,1/t_\ell)$,
  it follows from $\|x-y\|_\infty\le \|x-y\|_2\le 0.1$ and 
  $\|p^*-p^{t_\ell}\|_\infty\le 0.1$ and (\ref{eq:hehe1})
  that
$$
y_{i}-y_j\le p^{t_\ell}_{i}-p^{t_\ell}_j-(1-0.4)<p^{t_\ell}_{i}-p^{t_\ell}_j
$$
and thus, $y\notin \calP_{i}(p^{t_\ell},+1)$.

The other case where $x_{i}-p^*_{i}<p^*_j-x_j$ is similar.
  \end{proof}
  
  Notice that $B(x,1/t_{\ell})$ are all disjoint from each other, so for every $x\in T$, we have 
  $$\frac{\vol\left(\left(\bigcup_{i\in[k]} \calP_{i}(p^{t_\ell},s_i)\right)\cap B(x,1/t_\ell)\right)}{\vol(S^{t_\ell})}\leq \frac{\vol(B(x,1/t_\ell))}{\vol(S^{t_\ell})}= \frac{1}{|T|}.$$
  
  Thus, Claim~\ref{claim: keykeykey} implies that 
  
\begin{equation}\label{eq: aaabbbccc}
\frac{\left|\left(\bigcup_{i\in[k]}\calP_i(p^*,s_i)\right)\cap T\right|}{|T|}\ge \frac{\vol\left(\left(\bigcup_{i\in[k]} \calP_{i}(p^{t_\ell},s_i)\right)\cap S^{t_{\ell}}\right)}{\vol(S^{t_\ell})}.
\end{equation}

It only remains to note that $$\vol\left(\left(\bigcup_{i\in[k]} \calP_{i}(p^{t_\ell},s_i)\right)\cap S^{t_{\ell}}\right)=\sum_{i\in[k]}\vol\left(\calP_{i}(p^{t_\ell},s_i)\cap S^{t_{\ell}}\right)$$ and $$\vol(\calP_{i}(p^{t_\ell},+1)\cap S^{t_\ell})=\vol(\calP_{i}(p^{t_\ell},-1)\cap S^{t_\ell})$$ for all $i\in[k]$. 
The first equation holds because the intersection of any two different pyramids $\calP_{i}(p^{t_\ell},s_i)$ lies
on a hyperplane and thus has 0 volume, and the second equation follows from the choice of  $p^{t_\ell}$.
So \Cref{eq: aaabbbccc} implies  \Cref{lemma: balanced 111} since 

$$
\frac{\left|\left(\bigcup_{i\in[k]}\calP_i(p^*,s_i)\right)\cap T\right|}{|T|}\ge \frac{\vol\left(\left(\bigcup_{i\in[k]} \calP_{i}(p^{t_\ell},s_i)\right)\cap S^{t_{\ell}}\right)}{\vol(S^{t_\ell})}=\sum_{i\in[k]}\frac{\vol(\calP_{i}(p^{t_\ell},s_i)\cap S^{t_\ell})}{\vol(S^{t_\ell})} = \frac{1}{2}.
$$

Now we consider the general case when $A$ is not empty and let $\alpha>0$ be the smallest value in $A$; note that $\alpha\le 1/2$.
In this case we let $\ell$ be a sufficiently large integer such that 
  $1/t_\ell\le 0.1\alpha$ and $\|p^{t_\ell}-p^*\|_\infty\le 0.1\alpha$.
Similarly it suffices to show that for every
  point $x\in T$ and $i\in[k]$, the following property holds 
$$
\calP_{i}(p^{t_\ell},s_i)\cap B(x,1/t_\ell)\ne \emptyset\ \ \Longrightarrow\ \ 
x\in \calP_{i}(p^*,s_i).
$$
Let's prove the contrapositive so take any $x\in T$ and $i\in[k]$ such that $x\notin \calP_{i}(p^*,s_i)$. Assume without loss of generality that $s_i=+1$.
There exists a $j\ne i$ such that either
$$
x_{i}-p^*_{i}<x_{j}-p^*_j\quad\text{or}\quad 
x_{i}-p^*_{i}<p^*_j-x_j.
$$
For the first case we have $x_{i}-x_j<p^*_{i}-p^*_j$. Since 
  $p^*_{i}-p^*_j$ is either an integer or an integer $\pm$ something that is between $\alpha$ and $1-\alpha$, we have
$$
x_{i}-x_j\le p^*_{i}-p^*_j-\alpha.
$$
The rest of the proof is similar.
\end{proof}

Given $p^*$, we round it to an integer point $q^*\in [0:n]^k$ as follows: 
Let $q_{i}^*\in [0,n]$ be an integer such that $|p^*_{i}-q^*_{i}|\le 1/2$ for every $i\in[k]$ (in the case that $q^*_{i}$ is not unique, we break ties by choosing $q^*_{i}\in\set{p^*_{i}-1/2,p^*_{i}+1/2}$ such that $q^*_{i}$ is odd). It is clear that $q^*\in[0:n]^k$.
  
\begin{lemma}\label{lemma: balanced 222}
The point $q^*$ satisfies $q^*\in [0:n]^k$ and 
$$
\calP_{i}(p^*,+1)\cap T\subseteq \calP_{i}(q^*,+1)\cap T\quad\text{and}\quad \calP_{i}(p^*,-1)\cap T\subseteq 
\calP_{i}(q^*,-1)\cap T \quad\text{for all } i\in[k].
$$
\end{lemma}
\begin{proof}
Fix an arbitrary $i\in[k]$. We show $\calP_{i}(p^*,+1)\cap T\subseteq \calP_{i}(q^*,+1)\cap T$ below, and the other part $\calP_{i}(p^*,-1)\cap T\subseteq 
\calP_{i}(q^*,-1)\cap T$ can be handled in the same way. 

Let $x\in \EVEN(n,k)$ be a point in $\calP_{i}(p^*,+1)\cap T$. So for every $j\ne i$ we have 
$$
x_{i}-p^*_{i}\ge x_{j}-p^*_j\quad\text{and}\quad 
x_{i}-p^*_{i}\ge p^*_j-x_j,
$$
or equivalently
$$
x_{i}- x_{j}\ge p^*_{i} -p^*_j\quad\text{and}\quad 
x_{i}+x_j\ge p^*_{i}+ p^*_j. 
$$
It suffices to show that 
$$
x_{i}-x_j\ge q^*_{i}-q^*_j\quad\text{and}\quad 
x_{i}+x_j\ge q^*_i+q^*_{j}.
$$
Note that 
$$
p_{i}^*\geq q^*_{i}-1/2\quad\text{and}\quad 
p_{j}^*\leq q_{j}^*+1/2.
$$
So we have
$$
x_{i}- x_{j}\ge p^*_{i} -p^*_j\geq (q^*_{i}-1/2)-(q^*_j+1/2)=(q^*_{i}-q^*_j)-1.
$$
Given the rounding rule for obtaining $q^*$ and the fact that $x\in \EVEN(n,k)$, we will argue that it must be the case where $x_{i}- x_{j}\ge q^*_{i}-q^*_j.$ Consider two cases below.

If 
$
p_{i}^*> q^*_{i}-1/2$ or $p_{j}^*< q_{j}^*+1/2.
$, then we know that 
$
x_{i}- x_{j}>q^*_{i}-q^*_j-1
$. Since both $x_{i}- x_{j}$ and $ q^*_{i}-q^*_j $ are integer numbers, this implies
$
x_{i}- x_{j}\ge q^*_{i}-q^*_j.
$

If 
$
p_{i}^*= q^*_{i}-1/2$ and $p_{j}^*= q_{j}^*+1/2,
$ then we know that both $q_{i}^*$ and $q_{j}^*$ are odd numbers, which means the parity of $|q_{i}^*-q_{j}^*|$ is even. Since the parity of $|x_{i}- x_{j}|$ is also even, $x_{i}- x_{j}\geq q^*_{i}-q^*_j-1$ implies  $x_{i}- x_{j}\geq q^*_{i}-q^*_j$.

The other part $x_{i}+x_j\ge q^*_i+q^*_{j}$ can be proved similarly.
\end{proof}
This finishes the proof of \Cref{lemma: existence of balanced point}.
\end{proof}

\noindent\textit{Remark on \Cref{lemma: existence of balanced point}.} We note that the (possibly off-grid point) $p^*$ in the proof already satisfies the desired property and our algorithm can proceed by querying $p^*$.
However, $p^*$ as defined here is the limit of fixed points found in an infinite sequence of maps.
In contrast, \Cref{lemma: balanced 222}  shows that, after rounding, a grid balanced point always exists, 
  which can be found by brute-force enumeration.

\section{Impossibility of Strong Approximation under Non-expansion}

We consider functions $f$ on the plane with bounded domain and range, e.g. the unit square, that are non-expansive under the $\ell_{\infty}$ metric. 
We will show the following impossibility result.

\impossiblestrongappxnonexpanding*

In the proof it will be more convenient to use as the domain a square that is tilted by
$45^\circ$. We call a rectangle whose sides are at $45^\circ$ and $-45^\circ$, a {\em diamond}. Let $D$ be the diamond whose vertices are the midpoints of the sides of the unit square. Any function $g$ over $D$ can be extended to a function $g'$ over the unit square, by defining for every point $p \in [0,1]^2$ the value of the function as  $g'(p) = g(\pi(p))$, where $\pi(p)$ is the projection of $p$ onto $D$, namely, the point obtained by translating $p$ along lines perpendicular to the sides of $D$ until it intersects the boundary of the diamond. Clearly, for any two points $p, q \in [0,1]^2$, $\norm{\pi(p) - \pi(q)} \leq \norm{p-q}$, hence if the function $g$ over $D$ is non-expansive, then so is the function $g'$ over $[0,1]^2$. Furthermore, the fixed points of $g'$ are exactly the fixed points of $g$.

We will prove the statement of the theorem for the domain $D$. 
The claim then follows for the unit square. To see this, restrict attention to the non-expansive functions $g'$ on the unit square that are extensions of functions $g$ on the diamond $D$. If we have an algorithm for the unit square, then we can use the algorithm also for the diamond $D$: when the algorithm queries a point $p \in [0,1]^2$ then we query instead its projection $\pi(q) \in D$. If the algorithm outputs at the end a point that is close to a fixed point of $g'$, then its projection on $D$ is a valid output for $g$.

For any $\delta \in (0,1/2)$ and any point $s$ on the SW or NE side of the diamond $D$ that is at least at Euclidean distance $\delta$ from the vertices of $D$, we will define a non-expansive function $f_{\delta,s}$ with unique fixed point $s$.
The function is defined as follows. Draw the line $l_0$ through $s$ at $45^\circ$ and let $t$ be the point of intersection with the opposite side of $D$. Let $l_1$ and $l_2$ be the two lines parallel to $l_0$ that are left and right of $l_0$ respectively at Euclidean distance $\delta$, and let $D_0$ be the strip of $D$ that is strictly between the lines $l_1$ and $l_2$. Let $D' = D \setminus D_0$.
Every point $p \in D'$ is mapped by $f_{\delta,s}$ to the point that is at Euclidean distance $\delta$ towards the line $l_0$; i.e., if $p=(p_1,p_2)$ is left and above $l_0$ then
$f_{\delta,s}(p) = (p_1 + \delta/\sqrt{2}, p_2 -\delta/\sqrt{2})$, and if $p$ is right and below $l_0$ then $f_{\delta,s}(p) = (p_1 - \delta/\sqrt{2}, p_2 +\delta/\sqrt{2})$.

For a point $p$ in $D_0$ we define $f_{\delta,s}(p)$ as follows. Let $p'$ be the projection of $p$ onto the line $l_0$. Then  $f_{\delta,s}(p)$ is the point on $l_0$
that is at Euclidean distance $(\delta-|pp'|) \cdot |p's|$ from $p'$ 
in the direction of $s$,
where $|pp'|, |p's|$ are the (Euclidean) lengths of the segments $pp'$ and $p's$.
Thus for example, if $p=s$ then $p'=s$ and $f_{\delta,s}(s)=s$. If $p=t$ then $p'=t$ and
$(\delta-|pp'|) \cdot |p's| = \delta /\sqrt{2}$, 
so $t$ moves along $l_0$ distance $\delta /\sqrt{2}$ towards $s$.
Note that if $p$ is on line $l_1$ or $l_2$ 
(i.e. on the boundary sides between $D_0$ and $D'$), then $(\delta-|pp'|) \cdot |p's|=0$,
since $|pp'| = \delta$, thus $p$ is mapped to $p'$ whether we treat $p$ as a member of $D'$ or as a member of $D_0$. It follows that $f_{\delta,s}(p)$ is continuous over $D$.

As we noted above, $s$ is a fixed point of $f_{\delta,s}(p)$. 
We claim that it is the only fixed point.
Clearly, any fixed point $p$ must be in $D_0$ and must lie on the line $l_0$,
thus $p=p'$. It must satisfy also $(\delta-|pp'|) \cdot |p's| =0$, hence $|p's|=0$,
and thus, $p=s$.

We will show now that $f_{\delta,s}$ is a non-expansive function,
i.e. that $\norm{ f_{\delta,s}(p) - f_{\delta,s}(q) } \leq \norm{ p-q }$ for all $p, q \in D$.
We show first that it suffices to check pairs $p,q$ that are {\em diagonal} to each other,
i.e. such that the line connecting them is at $45^\circ$ or $-45^\circ$.
Note that such points have the property that the $L_{\infty}$ distance is tight in
both coordinates, $\norm{ p-q } = |p_1 -q_1| = |p_2 - q_2|$. 

\begin{lemma} \label{lem:diagonal2}
If a function $f$ on the diamond $D$ satisfies  $\norm{ f(p) - f(q) } \leq \norm{ p-q }$ 
for all diagonal pairs of points $p, q$, then $f$ is non-expansive.
\end{lemma}
\begin{proof}
Let $x,y$ be any two points that are not diagonal. 
Consider the diamond with opposite vertices $x, y$,
i.e. draw the lines through $x, y$ at $45^\circ$ and $-45^\circ$
and consider the rectangle enclosed by them.
Let $z, w$ be the other two vertices of this diamond. 
Suppose without loss of generality that
$\norm{ x -y } = x_1 - y_1 > |x_2 - y_2|$.
Then $x_1 > z_1 > y_1$, and similarly for $w$.
We have $\norm{ x -y } = x_1 - y_1 = ( x_1 -z_1) + (z_1 - y_1)$
$= \norm{ x -z } + \norm{ z -y }$.
Since $f$ is non-expansive on diagonal pairs,
$\norm{ f(x) - f(z) } \leq \norm{ x-z }$ and $\norm{ f(z) - f(y) } \leq \norm{ z-y }$.
Therefore, $\norm{ f(x) - f(y) } \leq \norm{ f(x) - f(z) } + \norm{ f(z) - f(y) }$
$\leq  \norm{ x-z } + \norm{ z-y } = \norm{ x -y }$.
\end{proof}

\noindent {\em Remark.} The lemma can be shown to hold more generally in any dimension. That is, if $f: [0,1]^k \mapsto [0,1]^k$ has the property that
$\norm{ f(p) - f(q) } \leq \norm{ p-q }$ 
for all diagonal pairs of points $p, q$ 
(i.e. such that $ |p_i-q_i| = \norm{ p-q } \text{ for all }i \in [k]$),  then $f$ is non-expansive.

\begin{lemma}
The function $f_{\delta,s}$ is non-expansive.
\end{lemma}
\begin{proof}
The function $f_{\delta,s}$ was defined according to which region
of the domain $D$ a point lies in. There are three regions:
the part of $D'$ left of $l_1$, the middle region $D_0$, and the part of $D'$ right of $l_2$. It suffices to check the non-expansiveness for diagonal pairs of points
$p, q$ that lie in the same region.
If $p, q$ are both in the region left of $l_1$, or if they are both right of $l_2$,
then from the definition we have $\norm{ f(p) - f(q) } = \norm{ p-q }$.

So suppose  $p, q$ are both in $D_0$. Assume first that the line $pq$ has angle $45^\circ$, i.e. $pq$ is parallel to the line $l_0$. 
Then $\norm{ p-q } = |pq|/\sqrt{2}$.
Let $p', q'$ be the projections of $p, q$ on $l_0$, and let
$p''= f_{\delta,s}(p)$, $q''=f_{\delta,s}(q)$.
Then $|p'p''| = (\delta-|pp'|) \cdot |p's|$, $|q'q''| = (\delta-|qq'|) \cdot |q's|$.
Since $pq$ is parallel to $l_0$, $(\delta-|pp'|)=(\delta-|qq'|)$ and $|pq| = |p'q'|$.
Assume without loss of generality that $|p's| > |q's|$.
Then $|p''q''| = |p'q'| - (\delta-|pp'|)(|p's|-|q's|) \leq |p'q'| = |pq|$.
Since $\norm{ p-q } = |pq|/\sqrt{2}$ and $\norm{ f(p) - f(q) } = |p''q''|/\sqrt{2}$,
it follows that $\norm{ f(p) - f(q) } \leq \norm{ p-q }$.

Assume now that the line $pq$ has angle $-45^\circ$, 
i.e., $pq$ is perpendicular to $l_0$. 
Again $\norm{ p-q } = |pq|/\sqrt{2}$.
Now $p$ and $q$ have the same projection $p'=q'$ on $l_0$.
Let $p''= f_{\delta,s}(p)$, $q''=f_{\delta,s}(q)$.
We have $|p'p''| = (\delta-|pp'|) \cdot |p's|$, 
and $ |q'q''| = |p'q''| =(\delta-|qp'|) \cdot |p's|$.
Therefore, $|p''q''| = | (|pp'| - |qp'|) | \cdot |p's|$.
If $p, q$ are on the same side of $l_0$ then $| (|pp'| - |qp'|) | = |pq|$.
If $p, q$ are on opposite sides of $l_0$ then $| (|pp'| - |qp'|) | < |pq|$.
In either case, we have $|p''q''| \leq |pq| \cdot |p's| < |pq|$,
since $|p's| \leq |st| = 1/\sqrt{2}$.
Again, since $\norm{ p-q } = |pq|/\sqrt{2}$ and $\norm{ f(p) - f(q) } = |p''q''|/\sqrt{2}$,
it follows that $\norm{ f(p) - f(q) } \leq \norm{ p-q }$.
\end{proof}

We are ready now to prove the theorem.
Intuitively, if the given function is $f_{\delta,s}$ for some $s$ on the
NE or SW side of $D$ and some small $\delta$, then for an algorithm (deterministic
or randomized) to find a point that is within $L_{\infty}$ distance 1/4 of $s$,
it must ask a query within the central region $D_0$ around $s$,
because otherwise it cannot know whether the fixed point $s$ is 
on the NE or the SW side of $D_0$.  

\noindent{\em Proof of \Cref{thm:strong-appx}.}
Recall that binary search is an optimal algorithm for searching for an unknown
item in a sorted array $A$, both among deterministic and randomized algorithms.
If the array has size $N$, then any randomized comparison-based algorithm
requires expected time at least $\log N -1$ to look up
an item in the array whose location is not known.

Suppose there is a (randomized) algorithm $B$ that computes a point
that is within 1/4 of a fixed point of a non-expansive function $f$ over
the domain $D$ within a finite expected number $n$ of queries
(the expectation is over the random choices of the algorithm). 
We will show how to solve faster the array search problem.
Partition the diamond $D$ into $N = 2^{2n}$ strips by
drawing $N-1$ parallel lines at $45^\circ$, spaced at distance $1/(N \sqrt{2})$ from each other, between the NW and SE side of $D$.
Let $S_1, \ldots, S_N$ be the $N$ strips.
Fix a $\delta < 1/ (N 2\sqrt{2})$.
For each $x \in [N]$, let $s_x$ be the point on the SW side of $S_x$ at Euclidean distance $\delta$ from the S vertex of $S_x$, 
and $t_x$ the point on the NE side of $S_x$ at Euclidean
 distance $\delta$ from the N vertex.
Note that $\norm{ s_x - t_x } > 1/2$.
Let ${\calF}$ be the family of non-expansive functions  $\{ f_{\delta,s_x}, f_{\delta,t_x} ~| x \in [N] \}$.

Consider the execution of algorithm $B$ for a function $f \in \calF$.
Note that the central region $D_0$ for the functions $f_{\delta,s_x}$ and $f_{\delta,t_x}$
is contained in the strip $S_x$.
If $B$ queries a point $p$ in another strip $S_j$, the answer $f(p)$ only
conveys the information whether $j < x$ or $j > x$.
If some execution of $B$ returns a point $q$ without ever having queried any point
in $S_x$, then all the answers in the execution are consistent with both
$f_{\delta,s_x}$ and $f_{\delta,t_x}$.
Since $\norm{ s_x - t_x } > 1/2$, either $\norm{q - s_x } > 1/4$ or $\norm{ q - t_x } > 1/4$.
Therefore, a correct algorithm $B$ cannot terminate before querying a point
in the strip $S_x$ that contains the fixed point of the function.

We can map now the algorithm $B$ to an algorithm $B'$ for the problem
of searching for an item in a sorted array $A$ of size $N$.
A choice of an index $x$ in the array $A$ corresponds to a choice
of the strip $S_x$ that contains the fixed point of the function $f \in \calF$,
i.e. choosing one of $f_{\delta,s_x}$, $f_{\delta,t_x}$.
Since $B$ terminates in expected number $n$ of queries, it asks 
within $n$ steps a query within the strip $S_x$ of the fixed point,
hence the expected time of the algorithm $B'$ is at most
$n = \log N/2$, a contradiction.
\qed

\section{Promise Problem versus Total Search Version}

The problem $\Contraction(\eps,\gamma,k)$ is a \emph{promise problem}, where we want to compute an $\eps$-fixed point of a given function $f$ with promise that $f$ is a $(1-\gamma)$-contraction. For a promise problem, one can define its total search version by asking to find a desired solution as in the promise problem, or a short violation certificate indicating that the given function doesn't satisfy the promise. The total search version of $\Contraction(\eps,\gamma,k)$,
denoted $\TContraction(\eps,\gamma,k)$ is naturally defined as the following search problem.

\begin{definition}[Total search version $\TContraction(\eps,\gamma,k)$]
	Given a function $f:[0,1]^k\mapsto [0,1]^k$, find one of the following:
	\begin{itemize}
	\item a point $x\in[0,1]^k$ such that $\norm{f(x)-x}\leq \eps$;
	\item two points $x,y\in[0,1]^k$ such that $\norm{f(x)-f(y)}>(1-\gamma)\norm{x-y}$.
	\end{itemize}
	In the black-box setting, the function $f$ is given by an oracle access. 
\end{definition}

Our theorem in this section shows $\TContraction(\eps,\gamma,k)$ admits the same query bounds as $\Contraction(\eps,\gamma,k)$.

\totalversion*

	\Cref{thm: total search version} follows from \Cref{lem:total} below.
	
	\begin{lemma}\label{lem:total}
Let $\set{q^1,\cdots,q^m}$ be a set of points in $[0,1]^k$ and $\set{a^1,\cdots,a^m}$ be the corresponding answers from the black-box oracle. There is a $(1-\gamma)$-contraction $f$ that is consistent with all the answers if and only if there is no pair $t_1,t_2$ such that $\norm{a^{t_1}-a^{t_2}}>(1-\gamma)\norm{q^{t_1}-q^{t_2}}$.
\end{lemma}
\begin{proof}
If there is some pair $t_1,t_2$ such that $\norm{a^{t_1}-a^{t_2}}>(1-\gamma)\norm{q^{t_1}-q^{t_2}}$, then obviously there is no $(1-\gamma)$ contraction that is consistent with the answers. 

Now suppose that no such pair exists. We define a function $f:[0,1]^k\mapsto [0,1]^k$ as follows: For every point $x\in[0,1]^k$ and coordinate $i\in[k]$, we let $f(x)_i=\min_{t\in[m]}\set{(1-\gamma)\norm{x-q^{t}}+a^t_i}$; if the minimal value of this set is larger than $1$, then we set $f(x)_i=1$. 

We show first that $f$ is consistent with all the query answers,
i.e. $f(q^j) = a^j$ for all $j \in [m]$.
Since the query points satisfy the contraction property,
we have $\norm{a^{j}-a^{t}} \leq (1-\gamma)\norm{q^{j}-q^{t}}$ for all $t \neq j$.
Therefore, for every coordinate $i \in [k]$,  $a^j_i \leq (1-\gamma)\norm{q^{j}-q^{t}} + a^t_i$.
Hence, $f(q^j)_i = \min_{t\in[m]}\set{(1-\gamma)\norm{q^j-q^{t}}+a^t_i} = a^j_i$.
Thus, $f(q^j) = a^j$.

We show now that the function $f$ constructed above is a $(1-\gamma)$-contraction. Consider any two points $x,y\in[0,1]^k$ and a coordinate $i\in[k]$. 
Suppose without loss of generality that $f(y)_i \leq f(x)_i$.
If $f(y)_i =1$, then also $f(x)_i=1$ and $|f(x)_i - f(y)_i| =0 \leq (1-\gamma)\norm{x-y}$.
So suppose $f(y)_i = (1-\gamma)\norm{y-q^{t}}+a^t_i$ for some $t \in [m]$.
By the triangle inequality,
$\norm{x-q^{t}} \leq \norm{x-y} + \norm{y-q^t} $.
Hence $f(x)_i \leq (1-\gamma)\norm{x-q^{t}} + a^t_i $
$\leq (1-\gamma)(\norm{x-y} + \norm{y-q^{t}}) + a^t_i= (1-\gamma)\norm{x-y}+ f(y)_i$.
Therefore, $0 \leq f(x)_i -f(y)_i \leq (1-\gamma)\norm{x-y}$.
Thus, $\norm{f(x) -f(y)}  \leq (1-\gamma)\norm{x-y}$.
\end{proof}

It follows from \Cref{lem:total} that we can use any algorithm that can solve the promise problem $\Contraction(\eps,\gamma,k)$ to solve the
total search version $\TContraction(\eps,\gamma,k)$ within the same number of queries (without loss of generality, we can assume that the algorithm for the promise problem $\Contraction(\eps,\gamma,k)$ always queries the approximate fixed point it finds).
To see it, consider the sequence of query points $q^1, \ldots, q^m$ generated by the algorithm and the corresponding
answers $a^1, \ldots, a^m$ of the black-box oracle.
If $\norm{q^i - a^i} \leq \epsilon$ for some $i$, then the algorithm will return the
$\epsilon$-approximate fixed point $q^i$ and terminate.
If there is a pair of generated queries that violate the contraction property,
then the algorithm can return the violating pair and terminate.
Suppose that neither of these events happens;
since all pairs of queries satisfy the contraction property,
there is a $(1-\gamma)$-contraction map $f$ that is consistent with all the queries.
Then the algorithm for the promise problem $\Contraction(\eps,\gamma,k)$ fails
to find an $\epsilon$-approximate fixed point for the black-box oracle function $f$
within the prescribed time bound, a contradiction.
Therefore, given any black-box oracle, the algorithm will find either an $\epsilon$-approximate fixed point
or a violating pair within the same number of queries $m=O(k \log (1/\epsilon))$ as the promise problem.

\section{Conclusions and Discussions}
We gave an algorithm for finding an $\eps$-fixed point of a contraction
(or non-expansive) map $f: [0,1]^k \mapsto [0,1]^k$ under the $\ell_{\infty}$ norm in polynomial query complexity.
Contraction maps under the $\ell_{\infty}$ norm are especially important because several  longstanding open problems from various fields can be cast in this framework.
The main open question is whether there is a general-purpose (black-box) algorithm for contraction maps that finds an approximate fixed point in polynomial time.
Resolving positively this question would have tremendous implications.

Our results provide the necessary first step towards the goal above since, in the black-box model, any polynomial-time algorithm must imply polynomial-query complexity. Nevertheless, we currently see obstacles to making our algorithm time-efficient. Namely, it is not clear how to efficiently find (or even verify) an approximate balanced point, whose existence guarantee is based on non-constructive arguments.

On the other hand, it is possible that finding an approximate fixed point of contraction maps is computationally hard (e.g., $\CLS$-hard or $\UEOPL$-hard). However, our result restricts what approaches could work in this direction. At a high level, such reduction must work in a white-box manner. To convey the interpretation more concretely, take the canonical $\CLS$-complete problem $\textsc{Continuous-Localopt}$ as an example. The problem takes two continuous functions, $p$ and $g$, as inputs and asks for an approximate local optimum of $p$ with respect to $g$; see its formal definition in \cite{FGHS23}. It could be the case where the reduction takes the polynomial-size arithmetic circuits that represent $p$ and $g$ as inputs and simulates their gates to generate a contraction map $f$ such that any approximate fixed point of $f$ encodes an approximate local optimum of $p$ with respect $g$. In contrast, it cannot be the case where the reduction treats the arithmetic circuits for $p$ and $g$ as evaluation oracles. Specifically, it cannot produce a contraction map $f$ where the evaluation of $f$ requires only a polynomial number of calls to the oracles for $p$ and $g$, since there are exponential query lower bounds for the $\textsc{Continuous-Localopt}$ problem \cite{hubacek2020hardness}.

An open question regarding the query complexity of $\Contraction(\eps,\gamma,k)$ is whether the upper bound $O(k\log(1/\eps))$ is tight. Currently, we are not aware of lower bounds that are stronger than the standard $\Omega(\log(1/\eps))$, which holds even when $k=1$. In an earlier version of this article, we asked the question of whether similar query bounds can be obtained for contraction maps under the $\ell_1$-norm or norms $\ell_p$ with $p>2$. A very recent unpublished paper~\cite{haslebacher2025query} answered this question affirmatively by showing $O(k^2\log(1/\eps\gamma)))$ query upper bounds for all norms $\ell_p$ with $p\in[1,\infty)\cup\set{\infty}$.

\section*{Acknowledgements} The authors are very grateful to the editor and anonymous reviewers for their helpful feedback
on improving the presentation.

\begin{flushleft}
\bibliographystyle{alpha}
\bibliography{ref}
\end{flushleft}

\end{document}